\documentclass{ifacconf}
\usepackage{graphicx} 
\graphicspath{{./images/}} 
\usepackage{mathtools} 
\usepackage{amsmath} 
\usepackage{amssymb} 
\usepackage{comment}
\usepackage{xcolor}
\usepackage{accents}
\usepackage{booktabs}

\usepackage{natbib}  

\newtheorem{remark}{Remark}
\newtheorem{assumption}{Assumption}
\newtheorem{definition}{Definition}

\usepackage{algorithm} 
\usepackage{algpseudocode} 

\newcommand{\set}[1]{\mathcal{#1}}
\newcommand{\ubar}[1]{\underaccent{\bar}{#1}}

\begin{document}
\begin{frontmatter}

\title{Safe Adaptive-Sampling Control via Robust $M$-Step Hold Model Predictive Control\thanksref{footnoteinfo}} 

\thanks[footnoteinfo]{This work is sponsored by the Department of the Navy, Office of Naval Research ONR N00014-24-2099.}

\author[UCB]{Spencer Schutz} 
\author[UCSB]{Charlott Vallon} 
\author[UCB]{Francesco Borrelli}

\address[UCB]{University of California, Berkeley, Berkeley, CA 94720 \\ 
    (e-mail: \{spencer.schutz, fborrelli\}@berkeley.edu).}

\address[UCSB]{University of California, Santa Barbara, Santa Barbara, CA 93106 
    (e-mail: cvallon@ucsb.edu).}

\begin{abstract}      
In adaptive-sampling control, the control frequency can be adjusted during task execution. Ensuring that these changes do not jeopardize the safety of the system being controlled requires attention. We introduce robust $M$-step hold model predictive control (MPC) to address this. Our formulation provides robust constraint satisfaction for an uncertain discrete-time system model with a fixed sampling time subject to an adaptable multi-step input hold (referred to as $M$-step hold). We show how to ensure recursive feasibility of the MPC utilizing $M$-step hold extensions of robust invariant sets, and demonstrate how to 
enable safe adaptive-sampling control via the 
online selection of $M$. We evaluate the utility of the robust $M$-step hold MPC formulation in a cruise control example.
\end{abstract}


\end{frontmatter}


\section{Introduction}\label{sec:intro}
In robust model predictive control (MPC), real-world continuous-time constraints are formulated as discrete-time constraints subject to modeled uncertainty (\cite{bemporad2007robust}). 
Safety is achieved by ensuring robust ``recursive feasibility" of the MPC optimization problem, a guarantee that formulated constraints are robustly satisfiable at all future time steps when the system is in closed-loop with the MPC (\cite{ciocca2017strong, fleming2013regions}). 
Robust recursive feasibility of a closed-loop MPC can be ensured by utilizing robust control invariant sets (\cite{kerrigan2000robust, yu2010robust, rakovic2010parameterized, schafer2023scalable}). If the system state is in a robust control invariant set, the MPC is guaranteed to have a solution that is \textit{(i)} feasible and \textit{(ii)} keeps the system within the robust control invariant set, for all modeled uncertainty realizations. Thus the task of ensuring robust closed-loop safety can be decomposed as first calculating robust control invariant sets and then designing a controller that ensures the system remains within the set.

Numerical techniques for computing robust control invariant sets assume a fixed control frequency  (\cite{Blanchini_Set_Inv},~\cite{Borrelli_MPC_book}). 
``Adaptive-sampling" control varies the control frequency to leverage different benefits throughout the task. High control frequency enables frequent input updates, providing stability for complex tasks and uncertain environments. Lower control frequency means fewer input updates, reducing computation in simple tasks and environments. 
Literature on adaptive-sampling PID (\cite{Dorf_Adaptive_Sampling_PID}), LQR (\cite{Henriksson_Adaptive_Sampling_LQR}), and MPC (\cite{Xue_Adaptive_Sampling_MPC, Gomozov_Adaptive_Sampling_MPC}) 
neither address 
recursive feasibility guarantees nor robustness to uncertainty. 

Robust adaptive-sampling control must preserve constraint satisfaction despite changes in control frequency. 
Here we specifically consider the problem of extending the robust control invariant set method to adaptive-sampling control.
In \cite{Schutz_MSH}, we introduced ``robust $M$-step hold control invariant sets," sets wherein robust constraint satisfiability is guaranteed for systems subject to a multi-step input hold (``$M$-step hold").
Importantly, $M$-step hold control invariance guarantees constraint satisfiability at all time steps despite only receiving control input updates every $M$ time steps. 
This is in contrast to other held-input techniques such as move-blocking MPC~(\cite{Cagienard_Move-block_MPC, Gondhalekar_Safe_Move-block_MPC}) and multi-horizon MPC~(\cite{BehrunaniMultiHorizonMPC}) which require input updates at every time step to ensure recursive feasibility. 

In this work we propose the design of a robust $M$-step hold MPC for safe adaptive-sampling control in conjunction with $M$-step hold invariant sets that guarantee recursive feasibility. We demonstrate a tractable formulation for LTI systems with a simulated cruise control example.

 

\section{Problem Formulation}\label{sec:prob}
Consider a constrained discrete-time, linear time-invariant system model with state $x_t\in\mathbb{R}^n$ and input $u_t\in\mathbb{R}^m$, subject to uncertainty $w_t\in\mathbb{R}^o$ $(t\in\mathbb{N}_0,~M\in\mathbb{N}_+)$:
\begin{subequations}\label{eqn:dyn_const}
\begin{gather}
    x_{t+1}=Ax_t+Bu_t+Ew_t \label{eqn:dyn}\\
    x_t\in\set{X},~u_t\in\set{U}, ~w_t\in\set{W}_{\mathrm{mod}(t,M)}. \label{eqn:const}
\end{gather}
\end{subequations}
\noindent The uncertainty model $\set{W}_{\mathrm{mod}(t,M)}$ is chosen to capture the effects of disturbances and modeling error compared to a real (possibly nonlinear) system (\cite{STEIN20112626, tomlin2017, hur2019, VOELKER2013943, roy2010}). In particular, the subscript ``$\mathrm{mod}(t,M)$'' indicates that the chosen uncertainty bounds reset every $M$ time steps, corresponding to uncertainty propagation between state measurements that occur every $M$ steps.

\begin{assumption}
The uncertainty bounds grow monotonically such that $\set{W}_k \subseteq \set{W}_{k+1}$ for all $k \in \mathbb{N}_{0}$.\label{asm:w_nest}
\end{assumption} 

We consider a controller in the form of an $M$-step hold control policy, as defined next.
\begin{definition} \label{def:msh}
An \textit{\textbf{$M$-step hold control policy}} $\pi^M_t(\cdot)$ calculates an input update every $M$ time steps:
\begin{align}\pi^M_t(\cdot)&=u_t \text{ s.t } u_t=u_{t-1} \text{ if mod}(t,M)\neq0 \nonumber\\
&=u_{M\cdot\lfloor t/M\rfloor } \label{eqn:msh_ctrl}
\end{align}
\end{definition}
The map $t\mapsto M\cdot\lfloor t/M\rfloor$ refers each step $t$ to the start of the $M$-step hold containing $t$, where $\lfloor\cdot/\cdot\rfloor$ indicates floor division. Thus, an $M$-step hold policy applies a constant input between updates. 

We define  \textbf{\textit{robust $M$-step hold model predictive control (MPC)}} as a MPC which solves the following optimal control problem every $M$ time steps (i.e. at all $t$ s.t. mod$(t,M)=0$):
\begin{subequations}\label{eqn:croc_full}
\begin{align}
    \min_{\Pi^M_{t}(\cdot)}~&\bar{x}_{t+N|t}^\top P \bar{x}_{t+N|t} + \sum_{k=t}^{t+N-1} \bar{x}_{k|t}^\top Q \bar{x}_{k|t} +\bar{u}_{i|t}^\top R \bar{u}_{i|t}\\
    \text{s.t.}~~&x_{k+1|t}=Ax_{k|t}+Bu_{i|t}+Ew_{k|t}\\
    &\bar{x}_{k+1|t}=A\bar{x}_{k|t}+B\bar{u}_{i|t}\\
    &x_{k|t}\in\set{X},~x_{t+N|t}\in\set{X}_N\label{eqn:croc_term_const}\\
    &x_{t|t}=\bar{x}_{t|t}=x_t\\
    &u_{i|t}=\pi^M_{i|t}(x_{i|t})\in\set{U}\\
    &\bar{u}_{i|t}=\pi^M_{i|t}(\bar{x}_{i|t})\in\set{U}\\
    &\forall w_{k|t}\in\set{W}_{\mathrm{mod}(k,M)}\\
    &i=M\cdot\lfloor k/M\rfloor,~\forall k\in\{t,\dots,t+N-1\}\nonumber
\end{align}
\end{subequations}
where $\Pi^M_t(\cdot)=\{\pi^M_{t|t}(\cdot),\pi^M_{t+M|t}(\cdot),\dots,\pi^M_{t+N-M|t}(\cdot)\}$ is a set of $M$-step hold policies; note that \eqref{eqn:croc_full} solves for $\frac{N}{M}$ policies. 
The predicted state $x_{k|t}$ at step $k$ must satisfy all state constraints robustly. We denote with $\{\bar{x}_{k|t},\bar{u}_{i|t}\}$ the nominal state and corresponding input, respectively. The sets $\set{X}$, $\set{X}_N$, $\set{U}$, and $\set{W}_{\mathrm{mod}(k,M)}$ are polytopes, and matrices $P\succeq0$, $Q\succeq0$, and $R\succ0$. 

The robust $M$-step hold MPC policy solves~\eqref{eqn:croc_full} every $M$ steps and applies the first optimal $M$-step hold policy to~\eqref{eqn:dyn} for $M$ time steps before re-solving (see Fig.~\ref{fig:msh_mpc}):
\begin{equation}
    \pi^M_{t,MPC}(\cdot)=
    \pi^{M\star}_{i|i}(\cdot) \text{ from~\eqref{eqn:croc_full}}\label{eqn:mpc_pol}, ~i=M\cdot\lfloor t/M\rfloor
\end{equation}
\begin{assumption}\label{asm:N_M}
    $N$ is an integer multiple of $M$.
\end{assumption}

%

Adaptive-sampling control can be achieved by adjusting the value of $M$ in the optimization~\eqref{eqn:croc_full} and policy~\eqref{eqn:mpc_pol}.
In the following section we will:
\begin{enumerate}
    \item show how to design the terminal constraint $\set{X}_N$ in \eqref{eqn:croc_term_const} to ensure recursive feasibility of the resulting controller for a single choice of $M$,
    \item propose a method for overcoming the intractability of optimizing over policies $\Pi^M_t(\cdot)$, and
    \item describe how to safely adapt $M$ while maintaining safety guarantees.
\end{enumerate}


\begin{figure}
    \centering
    \includegraphics[width=1\linewidth]{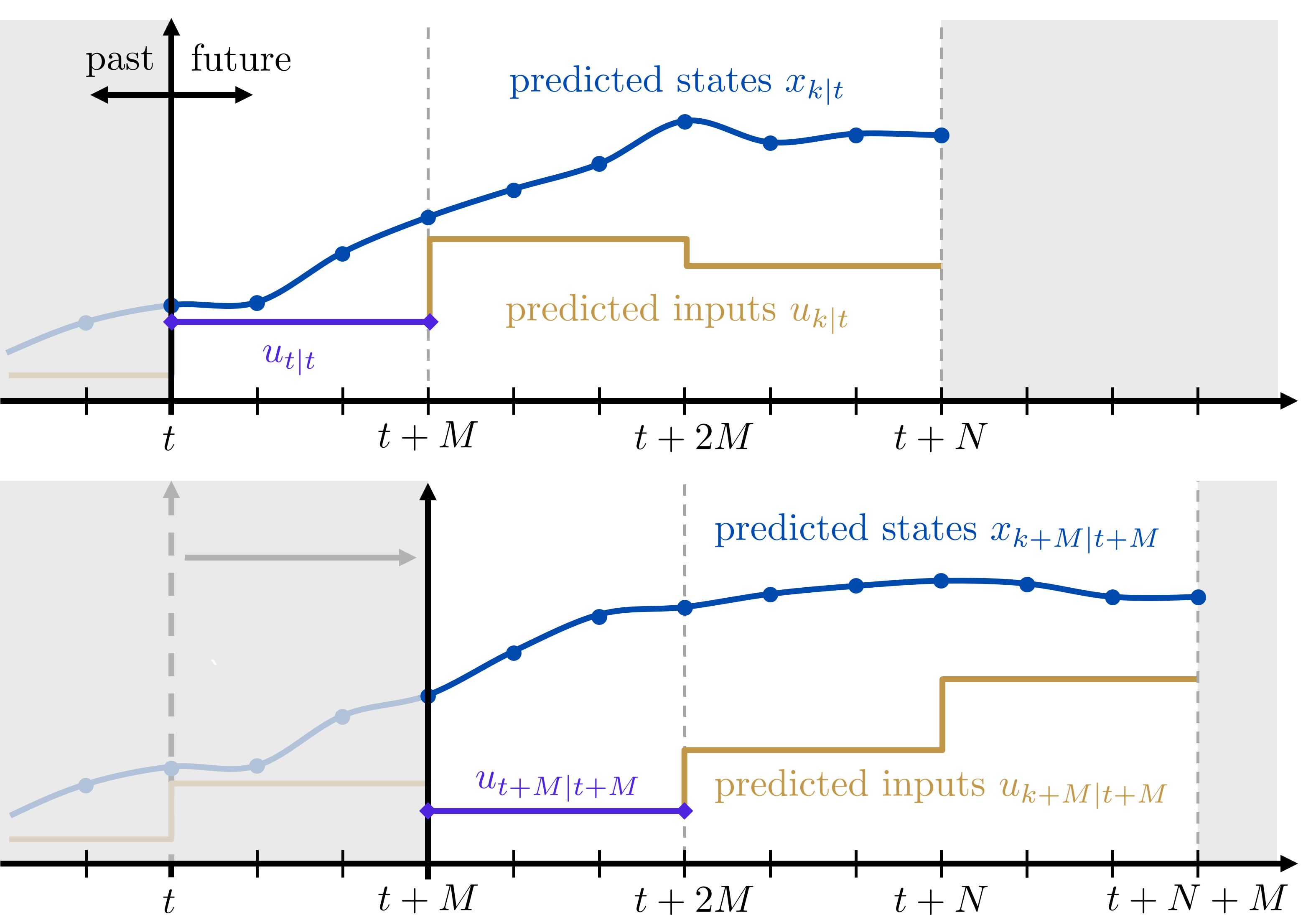}
    \caption{An $M$-step hold MPC is a receding horizon controller that solves an optimization every $M$ steps and applies a constant input between solves.}
    \label{fig:msh_mpc}
\end{figure}

\section{Robust \textit{M}-Step Hold MPC Design}\label{sec:mpc}
We present the design of a recursively feasible, robust $M$-step hold MPC~\eqref{eqn:croc_full}-\eqref{eqn:mpc_pol}.

\subsection{Robust \textit{M}-step Hold Invariance}

Proper selection of the MPC terminal constraint $\set{X}_N$ in \eqref{eqn:croc_term_const} will allow us to guarantee recursive feasibility. 
Our terminal constraint design will require the concept of robust $M$-step hold invariance, which extends traditional invariance to systems in closed-loop with an $M$-step hold. 
\begin{definition}\label{def:msh_ctrl_inv}
A set $\mathcal{C}^M\subseteq{\mathcal{X}}$ is \textit{\textbf{robust M-step hold control invariant}} for system~\eqref{eqn:dyn} with an $M$-step hold and subject to constraints~\eqref{eqn:const} if:
    \begin{align*}
    &x_t\in\set{C}^M \Rightarrow \exists~ u_{t}\in\set{U} \text{ s.t. } x_{t+1}=Ax_t+Bu_{t}+Ew_t,\\
    &~~~x_{t+1}\in\set{X},~x_{t+1}\in\set{C}^M~\text{if }\mathrm{mod}(t+1,M)=0,\\
    &~~~u_t=u_{M\cdot\lfloor t/M\rfloor},~\forall w_t\in\set{W}_{\mathrm{mod}(t,M)},~\forall t\in\mathbb{N}_0.
    \end{align*}
    
\end{definition}
The \textit{\textbf{maximal robust $M$-step hold control invariant set}} $\set{C}^M_\infty$ is the robust $M$-step hold control invariant set containing all $\set{C}^M\subseteq\mathcal{X}$ (see Appendix~\ref{apx:sets}). 

Robust $M$-step hold control invariance requires that the system~\eqref{eqn:dyn} be driven into the invariant set at least every $M$ time steps while robustly satisfying constraints~\eqref{eqn:const} at all time steps.
Periodic set re-entry has been previously defined (in absence of an $M$-step hold) as $k$-recurrence~(\cite{Shen_k-recurrence}) and $p$-invariance~(\cite{Olaru_p-invariance}), but these permit intermittent constraint violation.

While control invariant sets define where recursive feasibility is possible for \textit{some potential} control law, positive invariant sets define where a \textit{specified} feedback policy will be recursively feasible.
\begin{definition}\label{def:msh_pos_inv}
A set $\mathcal{O}^M\subseteq{\mathcal{X}}$ is \textit{\textbf{robust M-step hold positive invariant}} for system~\eqref{eqn:dyn} in closed-loop with a policy $\pi^M_{t}(\cdot)$ and subject to constraints~\eqref{eqn:const} if:
    \begin{align*}
    &x_t\in\set{O}^M \Rightarrow \exists \ u_t\in\set{U} \text{ s.t. } x_{t+1}=Ax_t+Bu_t+Ew_t,\\
    &~~~x_{t+1}\in\set{X},~x_{t+1}\in\set{O}^M~\text{if }\mathrm{mod}(t+1,M)=0,\\
    &~~~u_t=\pi^M_t(\cdot),~\forall w_t\in\set{W}_{\mathrm{mod}(t,M)},~\forall t\in\mathbb{N}_0.
    \end{align*}
\end{definition}
For a chosen policy $\pi^M_{t}(\cdot)$, the \textit{\textbf{maximal robust $M$-step hold positive invariant set}} $\set{O}^M_\infty$ is the robust $M$-step hold positive invariant set containing all $\set{O}^M\subseteq\mathcal{X}$ (see
Appendix~\ref{apx:sets}). Note that for any $\pi^M_{t}(\cdot)$, $\set{O}^M_\infty\subseteq\set{C}^M_\infty$.

We also introduce robust $M$-step hold controllability, which will be used to analyze the set of feasible initial states $\set{X}_0$ for~\eqref{eqn:croc_full} with a chosen $\set{X}_N$. 

\begin{definition} \label{def:msh_pre_set}
The \textbf{\textit{robust $M$-step hold precursor set}} $\mathrm{Pre}^M(\set{S},\set{W})$ to target set $\set{S}$ for system~\eqref{eqn:dyn} with an $M$-step hold and subject to constraints~\eqref{eqn:const} is
\begin{align*}
    &\mathrm{Pre}^M(\mathcal{S},\set{W}) = \big\{x_0:\exists \ u_0\in\mathcal{U} \text{ s.t. }\nonumber \\
    &~~~x_{t+1}=Ax_t+Bu_0+Ew_t,~x_{t+1}\in\set{X},~x_M\in\set{S},\nonumber\\
    &~~~\forall w_t\in\set{W}_t,~\forall t\in\{0,\dots,M-1\}\big\} .
\end{align*}
\end{definition}
Computation of $\mathrm{Pre}^M(\set{S},\set{W})$ is discussed in Appendix~\ref{apx:sets}. 
\begin{definition} \label{def:msh_ctrb_set}
The \textbf{\textit{robust $N$-step $M$-step hold controllable set}} $\set{K}^M_N(\set{S})$ to target $\set{S}$ is defined recursively as
\begin{align*}
    \set{K}^M_{i
    }(\set{S})&=\mathrm{Pre}^M(\set{K}^M_{i-M}(\set{S}),\set{W})\cap\set{X}\\
    \set{K}^M_0(\set{S})&=\set{S}, ~\forall i\in\{M,2M\dots,N\}.
\end{align*}
\end{definition}
For any choice of $\set{X}_N$ in~\eqref{eqn:croc_full}, the corresponding $\set{X}_0=\set{K}^M_N(\set{X}_N)$. By Def.~\ref{def:msh_ctrb_set}, all $x_0\in\set{X}_{0}$ may be driven to $\set{X}_N$ in $N$ steps by a sequence of $\frac{N}{M}$ $M$-step holds while satisfying constraints at all steps for any disturbance realization.

\begin{remark}
Defs.~\ref{def:msh_ctrl_inv}-\ref{def:msh_ctrb_set} describe a relaxed notion of robust $M$-step hold invariance compared to~\cite{Schutz_MSH}. The original notion enforced set containment at all time steps, not every $M$ time steps. 
\end{remark}

\subsection{Tractable Formulation}\label{subsec:rob}
Optimizing~\eqref{eqn:croc_full} over policies is intractable, so we reformulate to optimize over inputs while preserving $\set{X}_0$. Inspired by the approach in~\cite{monimoy_robust_mpc}, we reformulate~\eqref{eqn:croc_full} to search for inputs that \textit{(i)} robustly satisfy open-loop constraints from $x_{t|t}$ to $x_{t+M|t}$ and \textit{(ii)} ensure that we can robustly reach the terminal set $\mathcal{X}_N$ from $x_{t+M|t}$. 
Note point \textit{(ii)} is in contrast with~\eqref{eqn:croc_term_const}, which constrains $x_{t+N|t}$ to be in the terminal set. 
We reformulate as:
\begin{subequations}\label{eqn:croc_min}
\begin{align}
    \min_{U^M_{t}}~&\bar{x}_{t+N|t}^\top P \bar{x}_{t+N|t} + \sum_{k=t}^{t+N-1} \bar{x}_{k|t}^\top Q \bar{x}_{k|t} +\bar{u}_{i|t}^\top R \bar{u}_{i|t}\\
    \text{s.t.}~&\bar{x}_{k+1|t}=A\bar{x}_{k|t}+B\bar{u}_{i|t}\\  
    &\bar{x}_{t+h|t}\in\set{X}\ominus \set{E}_{h}\label{eqn:croc_min_X}\\
    &\bar{x}_{t+M|t}\in\set{K}^M_{N-M}(\set{X}_N)\ominus\set{E}_{M}\label{eqn:croc_min_K}\\
    &\bar{x}_{k|t}\in\set{X}\label{eqn:croc_min_nom}\\
    &\bar{x}_{t|t}=x_t\\
    &\bar{u}_{i|t}\in\set{U},~ i=M\cdot\lfloor k/M\rfloor\\
    \forall h\in&\{1,\dots,M-1\},~\forall k\in\{t,\dots,t+N-1\}\nonumber
\end{align}
\end{subequations}
where $\mathcal{X}_N$ is a robust $M$-step hold invariant set, and 
$\set{E}_k$ is the $k$-step uncertainty forward reachable set:
\begin{equation}\label{eqn:Ek}
    \mathcal{E}_k=\bigoplus_{j=0}^{k-1}A^{k-1-j}E\circ \mathcal{W}_{\mathrm{mod}(j,M)}
\end{equation}
We use $\ominus$ in~\eqref{eqn:croc_min_X}-\eqref{eqn:croc_min_K} to denote the Pontryagin difference, and $\oplus$ and $\circ$ in~\eqref{eqn:Ek} to denote the Minkowsi sum and image, respectively (\cite{Borrelli_MPC_book}). These operators are defined as
\begin{subequations}
    \begin{align}
        \mathcal{A}\ominus\mathcal{B}&=\{x:x+b\in\mathcal{A},~\forall b\in\mathcal{B}\},\label{eqn:pontryagin}\\
        \mathcal{A}\oplus\set{B}&=\{x+b:x\in\set{A},~b\in\set{B}\},\label{eqn:minkowski}\\
        Q\set{A}&=\{Qa: a\in\set{A}\}.\label{eqn:image}
    \end{align}
\end{subequations}

The variables $U^M_t=\{\bar{u}_{t|t},\bar{u}_{t+M|t},\dots,\bar{u}_{t+N-M|t}\}$ are $\frac{N}{M}$ input values, not policies.
Expressing~\eqref{eqn:croc_min_X}-\eqref{eqn:croc_min_K} in terms of (precomputed) sets lets~\eqref{eqn:croc_min} be written only in terms of the nominal states and inputs $\{\bar{x}_{k|t},\bar{u}_{i|t}\}$, and avoids the exponential blow-up of vertex enumeration methods~(\cite{Borrelli_MPC_book}). Constraints~\eqref{eqn:croc_min_X}-\eqref{eqn:croc_min_K} enforce robust open-loop feasibility for $M$ steps. 
By Defs.~\ref{def:msh_pre_set}-\ref{def:msh_ctrb_set}, $\set{X}_0$ for ~\eqref{eqn:croc_min} is $\mathrm{Pre}^M(\set{K}^M_{N-M}(\set{X}_N),\set{W})=\set{K}^M_N(\set{X}_N)$, the same as that of~\eqref{eqn:croc_full}. In the following section, we will show that choosing a robust $M$-step hold invariant $\set{X}_N$ guarantees~\eqref{eqn:croc_min} has a solution every $M$ steps. 

The robust $M$-step hold MPC policy solves~\eqref{eqn:croc_min} every $M$ steps, each time applying the first optimal $M$-step hold input to~\eqref{eqn:dyn} for $M$ time steps before re-solving~\eqref{eqn:croc_min}:

\begin{equation}
    \pi^M_{t,MPC}(\cdot)=
    \bar{u}^{\star}_{i|i} \text{ from~\eqref{eqn:croc_min}}\label{eqn:mpc_pol_new}, ~i=M\cdot\lfloor t/M\rfloor
\end{equation}

All sets in~\eqref{eqn:croc_min} may be precomputed offline, and every online evaluation of~\eqref{eqn:croc_min} only requires solving a single convex quadratic program. This tractable formulation is used in the example in Sec.~\ref{sec:ex}.


\subsection{Terminal Constraint}\label{subsec:const}

We can ensure recursive feasibility of~\eqref{eqn:dyn_const} in closed-loop with~\eqref{eqn:mpc_pol_new} by selecting a robust $M$-step hold (positive or control) invariant terminal set $\set{X}_N$.
\begin{thm}\label{thm:rec_feas}

    Consider an uncertain system~\eqref{eqn:dyn_const} in closed-loop with a robust $M$-step hold MPC policy~\eqref{eqn:croc_min}-\eqref{eqn:mpc_pol_new}. Assume the terminal set  $\set{X}_N$ in~\eqref{eqn:croc_min} is a robust $M$-step hold invariant set for the system~\eqref{eqn:dyn_const} according to Def.~\ref{def:msh_ctrl_inv} or Def.~\ref{def:msh_pos_inv}. Let Asm.~\ref{asm:N_M} hold. Assume $\set{E}_k$ contains the origin. If~\eqref{eqn:croc_min} is feasible at a time step $t=0$, then the robust $M$-step hold MPC policy~\eqref{eqn:croc_min}-\eqref{eqn:mpc_pol_new} is feasible every subsequent $M$ time steps $\{t+M,~t+2M,\dots \}$. 
    Furthermore, the realized system trajectory robustly satisfies all state and input constraints~\eqref{eqn:const} at all $t\in\mathbb{N}_0$.
\end{thm}
\begin{pf}
We will show this using induction.
Let~\eqref{eqn:croc_min}-\eqref{eqn:mpc_pol_new} be feasible at some time $t$. 
The first input $\bar{u}^\star_{t|t}$ satisfies constraints~\eqref{eqn:croc_min_X}-\eqref{eqn:croc_min_K}, so the system's nominal dynamics evolve as
\begin{align*}
    &\bar{x}_{t+k|t}\in\mathcal{X}\ominus\mathcal{E}_k,~\forall k\in\{1,\dots,M-1\}\\
    &\bar{x}_{t+M|t}\in\mathcal{K}^M_{N-M}(\mathcal{X}_N)\ominus\mathcal{E}_M.
\end{align*} 
By definition of the Pontryagin difference~\eqref{eqn:pontryagin}, this implies
\begin{align*}
    &\bar{x}_{t+k|t}+e_{t+k}\in\mathcal{X},~\forall e_{t+k}\in\mathcal{E}_k,~\forall k\in\{1,\dots,M-1\}\\
    &\bar{x}_{t+M|t}+e_{t+M}\in\mathcal{K}^M_{N-M}(\mathcal{X}_N),~\forall e_{t+M}\in\mathcal{E}_M.
\end{align*}
The true dynamics evolve from $x_{t}=\bar{x}_{t|t}$ as 
\begin{equation*}x_{t+k}=\bar{x}_{t+k|t}+e_{t+k},~
\end{equation*}
where $e_{t+k}=\sum_{j=0}^{k-1}A^{k-1-j}Ew_j\in\mathcal{E}_k$, and thus constraints~\eqref{eqn:croc_min_X}-\eqref{eqn:croc_min_K} guarantee
\begin{align*}
    &x_{t+k}\in\mathcal{X},~\forall k\in\{1,\dots,M-1\}\\
    &x_{t+M}\in\mathcal{K}^M_{N-M}(\mathcal{X}_N).
\end{align*}
    
Problem~\eqref{eqn:croc_min} is re-solved at $x_{t+M}$ with initial condition $x_{t+M}=\bar{x}_{t+M|t+M}$.  By Def.~\ref{def:msh_ctrb_set}, 
\begin{equation*}
    x_{t+M}\in \mathrm{Pre}^M(\mathcal{K}^M_{N-2M}(\mathcal{X}_N),\mathcal{W})\cap\mathcal{X}.
\end{equation*}
By Def.~\ref{def:msh_pre_set}, this guarantees the existence of an $M$-step control input $u^\dagger$ that can be applied beginning at time $t+M$ such that 
\begin{align*}
    &x_{t+M+k}\in\mathcal{X},~\forall k\in\{1,\dots,M-1\}\\
    &x_{t+2M}\in\mathcal{K}^M_{N-2M}(\mathcal{X}_N)
\end{align*}
for all uncertainty realizations. We now show that this implies the constraints~\eqref{eqn:croc_min_X}-\eqref{eqn:croc_min_K} are satisfiable beginning at time $t+M$. Note that the true and nominal dynamics are related as 
\begin{equation*}
    x_{t+M+k}=\bar{x}_{t+M+k|t+M}+e_k 
\end{equation*}where $e_k\in\mathcal{E}_k$. This implies that the candidate solution $\bar{u}^\star_{t+M|t+M}=u^\dagger$ guarantees \begin{align*}
    &\bar{x}_{t+M+k|t+M}+e_{k}\in\mathcal{X},~\forall e_k\in\mathcal{E}_k,~\forall k\in\{1,\dots,M-1\}\\
    &\bar{x}_{t+2M|t+M}+e_{M}\in\mathcal{K}^M_{N-2M}(\mathcal{X}_N),~\forall e_M\in\mathcal{E}_M.
\end{align*}
By the Pontryagin difference~\eqref{eqn:pontryagin}, this implies that 
\begin{align*}
    &\bar{x}_{t+M+k|t+M}\in\mathcal{X}\ominus\mathcal{E}_k,~\forall k\in\{1,\dots,M-1\}\\
    &\bar{x}_{t+2M|t+M}\in\mathcal{K}^M_{N-2M}(\mathcal{X}_N)\ominus\mathcal{E}_M
\end{align*}

When $\mathcal{X}_N$ is robust $M$-step hold invariant, the set of states robustly $M$-step hold controllable to $\mathcal{X}_N$ in $N-2M$ steps is a subset of those controllable in $N-M$ steps:
\begin{equation*}
    \mathcal{K}^M_{N-2M}(\mathcal{X}_N)\subseteq\mathcal{K}^M_{N-M}(\mathcal{X}_N)
\end{equation*}
The Pontryagin difference has the property that if $\mathcal{A}\subseteq\mathcal{B}$ then $\mathcal{A}\ominus\mathcal{C}\subseteq\mathcal{B}\ominus\mathcal{C}$~(\cite{Kolmanovsky_sets}), so
\begin{equation*}
    \mathcal{K}^M_{N-2M}(\mathcal{X}_N)\ominus\mathcal{E}_M\subseteq\mathcal{K}^M_{N-M}(\mathcal{X}_N)\ominus\mathcal{E}_M. 
\end{equation*}
Thus, $\bar{x}_{t+2M|t+M}\in\mathcal{K}^M_{N-M}(\mathcal{X}_N)\ominus\mathcal{E}_M$ and the candidate $\bar{u}^\star_{t+M|t+M}=u^\dagger$ satisfies constraints~\eqref{eqn:croc_min_X}-\eqref{eqn:croc_min_K}.

Furthermore, $\bar{x}_{t+2M|t+M}\in\mathcal{K}^M_{N-M}(\mathcal{X}_N)\ominus\mathcal{E}_M$ implies that $\bar{x}_{t+2M|t+M}\in\set{K}^M_{N-2M}(\set{X}_N)$. By Def.~\eqref{def:msh_ctrb_set}, there exists a sequence of $(\frac{N}{M}-2)$ $M$-step holds that robustly drives the system from $\bar{x}_{t+2M|t+M}$ to $\set{X}_N$, and another $M$-step hold that robustly returns the system to the invariant $\set{X}_N$, all while staying within $\set{X}$. Since this robustness includes the uncertainty realization $w_t=0$, there exist $(\frac{N}{M}-1)$ $M$-step holds that keep the nominal $\bar{x}_{t+2M+1|t+M},\dots,\bar{x}_{t+N+M|t+M}\in\set{X}$. This guarantees the existence of a candidate input sequence after $\bar{u}^\star_{t+M|t+M}$ that satisfies constraint~\eqref{eqn:croc_min_nom}.

We have shown that if~\eqref{eqn:croc_min} is feasible at time $t$, then problem~\eqref{eqn:croc_min} has a solution at $t+M$ and $\bar{u}^\star_{t+M|t+M}$ guarantees constraint satisfaction for $M$ steps. By assumption,~\eqref{eqn:croc_min} was feasible at time $t=0$. We conclude by induction that~\eqref{eqn:croc_min} has a solution every $M$ steps and the MPC~\eqref{eqn:croc_min}-\eqref{eqn:mpc_pol_new} guarantees robust constraint satisfaction at every step.\qed
\end{pf}

\begin{remark} \label{rmk:rec_feas}
    The choice of $\set{X}_N$ between different robust $M$-step invariant sets impacts the size of $\set{X}_0=\set{K}^M_N(\set{X}_N)$.
    When $\set{X}_N$ is $M$-step hold (positive or control) invariant,~\eqref{eqn:mpc_pol_new} is recursively feasible. This implies 
    \begin{equation*}
    \set{X}_{0}=\set{K}^M_N(\set{X}_N)=\set{O}^M_{\infty,MPC}.
    \end{equation*}
 The largest possible $\set{O}^M_{\infty,MPC}$ is in fact $\set{C}^M_\infty$, obtained by choosing $\set{X}_N=\set{C}^M_\infty$ itself. Therefore, to maximize the safe operating region of~\eqref{eqn:croc_min}-\eqref{eqn:mpc_pol_new} we should select $\set{X}_N=\set{C}^M_\infty$. 
\end{remark}

\subsection{Safe Adaptive-Sampling Control}\label{subsec:adapt}
The $M$-step hold framework provides a convenient mechanism for safe adaptive-sampling control; specifically, we consider the safe online adaptation of $M$ at time steps $t$ when mod$(t,M)=0$.  

\begin{thm}\label{thm:change_M}
    Consider a system~\eqref{eqn:dyn_const} in closed-loop with a recursively feasible robust $M$-step hold MPC policy~\eqref{eqn:croc_min}-\eqref{eqn:mpc_pol_new}. Let Asm.~\ref{asm:N_M} hold for $M$ and some $\hat{M}\in \mathbb{N}_+$. From state $x_t$ at time $t$ when mod$(t,M)=0$, a robust $\hat{M}$-step hold MPC policy solving~\eqref{eqn:croc_min} with a robust $\hat{M}$-step hold invariant $\hat{\set{X}}_N$ will be recursively feasible if $x_t$ is in the set of feasible initial states $\set{\hat{X}}_0$: \begin{equation*}x_{t}\in\hat{\set{X}}_0=\set{K}^{\hat{M}}_{N}(\hat{\set{X}}_N)
    \end{equation*}
\end{thm}
\begin{pf}
    The robust $\hat{M}$-step hold MPC policy solving~\eqref{eqn:croc_min} is recursively feasible by Thm.~\ref{thm:rec_feas} if $\hat{\set{X}}_N$ is $\hat{M}$-step hold invariant and $x_t$ is a feasible initial condition. The latter is true if $x_t$ can be robustly driven into $\set{K}^{\hat{M}}_{N-\hat{M}}(\hat{\set{X}}_N)$ in $\hat{M}$ steps with a held input while robustly satisfying state constraints at all steps. By Defs.~\ref{def:msh_pre_set}-\ref{def:msh_ctrb_set}, the set of such states $\hat{\set{X}}_0$ is $\mathrm{Pre}^{\hat{M}}\big(\set{K}^{\hat{M}}_{N-\hat{M}}(\hat{\set{X}}_N)\big)\cap\set{X}=\set{K}^{\hat{M}}_N(\hat{\set{X}}_N)$.\qed
\end{pf}

\begin{remark}
    The $\hat{M}$-step hold policy in Thm.~\ref{thm:change_M} may use a new horizon $\hat{N}\neq N$ provided $\hat{M}$ and $\hat{N}$ satisfy Asm.~\ref{asm:N_M}.
\end{remark}

A feasible initialization by Thm.~\ref{thm:change_M} is guaranteed at any time step $t$ s.t. mod$(t,M)=0$ for all $\hat{M}$ that are factors of $M$ when the original $\set{X}_N=\set{C}^M_\infty$ and the new $\hat{\set{X}}_N=\set{C}^{\hat{M}}_\infty$. This is because maximal robust $M$-step hold control invariant sets have guaranteed evolution with respect to $M$, as formalized in Thm.~\ref{thm:nest}.

\begin{thm}\label{thm:nest}
    For a model~\eqref{eqn:dyn_const} and any $M,\hat{M}\in\mathbb{N}_+$, if $\hat{M}$ is a factor of $M$, then the maximal robust $M$-step hold control invariant set is a subset of the maximal robust $\hat{M}$-step hold control invariant set:
    \begin{gather*}
       \big(\mathrm{mod}(M,\hat{M}\big)=0\big) \Rightarrow
       \big(\set{C}^{M}_\infty \subseteq \set{C}^{\hat{M}}_\infty\big)
    \end{gather*}
\end{thm}
\begin{pf}
    We prove Thm.~\ref{thm:nest} by constructing a robust $\hat{M}$-step hold control invariant set from valid $M$-step hold trajectories. Let $M=q\hat{M}$ with $q\in\mathbb{N}_{+}$. Let $\Gamma\subseteq\set{C}^M_\infty \times \set{U}$ be the set of all state-input pairs that satisfy the robust $M$-step hold control invariance conditions of Def.~\ref{def:msh_ctrl_inv}. Let $\set{R}_i$ be the $i$-step robust forward reachable set of states from $\set{C}^M_\infty$ under the corresponding constant inputs in $\Gamma$ and the $\hat{M}$-periodic uncertainty sequence $\set{W}_{\mathrm{mod}(i,\hat{M})}$. 
    
    The modulo operator has the property $\mathrm{mod}(i,\hat{M})\leq i$ for all $i\in\mathbb{N}_{0}$, so by  Asm.~\ref{asm:w_nest}, $\set{W}_{\mathrm{mod}(i,\hat{M})}\subseteq\set{W}_i$. Consequently, the reachable sets under $\hat{M}$-periodic uncertainty are subsets of those under the original $M$-step uncertainty. Therefore, $\set{R}_0=\set{C}^M_\infty$, $\set{R}_i\subseteq\set{X}$ for all $i\in\{0,\dots,M-1\}$, and $\set{R}_M\subseteq\set{C}^M_\infty$. 
    
    We construct a candidate set
    \begin{equation}
        \Omega = \bigcup_{p=0}^{q-1}\set{R}_{p\hat{M}}.
    \end{equation}
    
    By definition of $\set{R}_i$, the final element $\set{R}_{(q-1)\hat{M}}$ robustly maps into $\set{R}_{q\hat{M}}=\set{R}_M$ under the inputs from $\Gamma$, which by construction is within $\set{R}_0$. This return property establishes that under valid $\hat{M}$-step hold inputs, $\Omega$ robustly maps into $\set{X}$ every step and into itself every $\hat{M}$ steps, rendering $\Omega$ robust $\hat{M}$-step hold control invariant. Maximality of $\set{C}^{\hat{M}}_\infty$ implies $\Omega\subseteq\set{C}^{\hat{M}}_\infty$, and thus $\set{C}^M_\infty\subseteq\set{C}^{\hat{M}}_\infty$. \qed
\end{pf}
Thms.~\ref{thm:change_M}-\ref{thm:nest} guarantee that the robust $M$-step hold MPC may safely adapt the hold size to any factor of $M$ at any time step $t$ s.t. mod$(t,M)=0$ when using maximal robust $M$-step hold control invariant terminal constraints.

Note that Thm.~\ref{thm:nest} does not necessarily apply if $\mathcal{X}_N$ in~\eqref{eqn:croc_min_K} is a maximal robust $M$-step hold \textit{positive} invariant set (Def.~\ref{def:msh_pos_inv}). 
In this  case, verifying $x_t\in\hat{\set{X}}_0$ is always necessary when switching control frequency, even if the hold duration is being adapted to a factor of $M$.

%

For both control invariant and positive invariant terminal sets, switching to an $\hat{M}>M$ requires first verifying that $x_t\in\hat{\set{X}}_0$. This is intuitive, as larger $M$-step holds must be robust to more open-loop uncertainty propagation and this may not always be feasible at the current state. If $x_t\notin\hat{\set{X}}_0$, the system must be steered into $\hat{\set{X}}_0$ before the switch can safely occur.
Any required pre-switch steering may be performed by adjusting the gains or reference in~\eqref{eqn:croc_min}. 

\section{Example: Cruise Control}\label{sec:ex}
We demonstrate robust $M$-step hold MPC in a cruise control for the first vehicle of a platoon.
\subsection{System and Constraints}
\begin{figure}[t]
    \centering
    \includegraphics[width=1\linewidth]{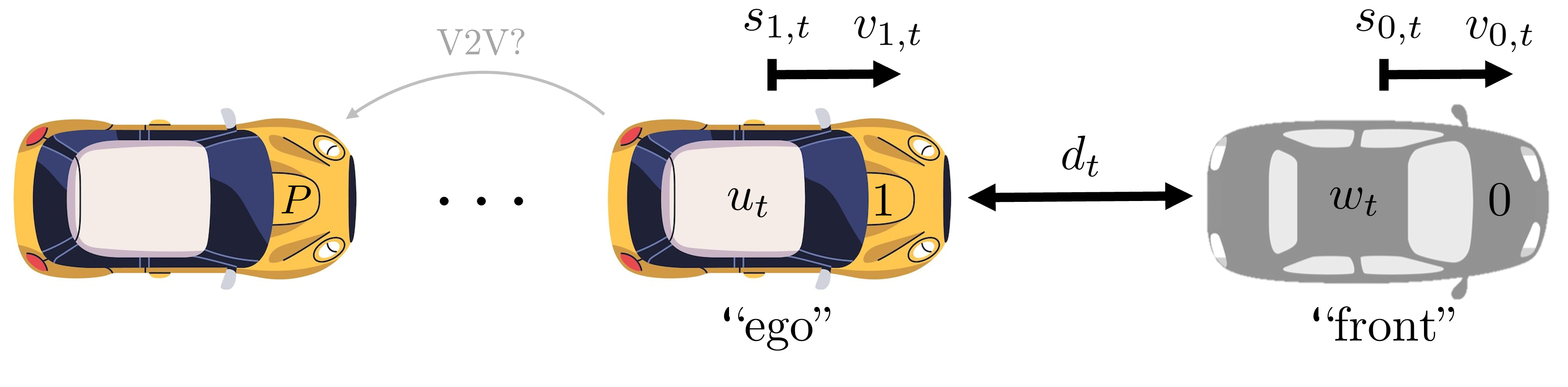}
    \caption{A vehicle platoon behind an uncontrolled front car.}
    \label{fig:platoon}
\end{figure}

Consider the ``front" ($0$) and ``ego" ($1$) cars in Fig~\ref{fig:platoon}, each with a longitudinal position in a road-aligned coordinate $\{s_{0,t},s_{1,t}\}$ and longitudinal velocity $\{v_{0,t},v_{1,t}\}$.
The input $u_t$ and uncertainty $w_t$ are the ego and front accelerations, respectively. Each car is modeled as a double integrator with sampling time $T_s$. For state $x_t=[d_t,v_{1,t},v_{0,t}]^\top$ with $d_t=s_{0,t}-s_{1,t}$,
\begin{align}\label{eqn:acc_dyn}
    x_{t+1}
    &= \begin{bmatrix} 1 & -T_s & T_s \\ 0 & 1 & 0 \\ 0 & 0 & 1 \end{bmatrix} x_t+\begin{bmatrix} -\frac{1}{2}T_s^2 \\ T_s \\ 0 \end{bmatrix} u_t + \begin{bmatrix} \frac{1}{2}T_s^2 \\ 0 \\ T_s\end{bmatrix} w_t \nonumber\\
    &=Ax_t+Bu_t+Ew_t.
\end{align} 
State constraints enforce upper and lower bounds on the following distance and ego velocity.
\begin{equation} \label{eqn:acc_X}
\set{X}=\{x:d_{\mathrm{min}}\leq d\leq d_{\mathrm{max}},~v_{\mathrm{min}}\leq v_1\leq v_{\mathrm{max}}\}
\end{equation}
Input constraints capture physical actuator limits.
\begin{equation} \label{eqn:acc_U}
\set{U}=\{u:u_{\mathrm{min}}\leq u\leq u_{\mathrm{max}}\}
\end{equation}
Following~\cite{Kim_ACC_2019}, we model a switched uncertainty to enforce the ego velocity limits on the front car. 
\begin{align} \label{eqn:acc_W}
    \set{W}_t(v_{0,t})&=\{w_t \mid \ubar{w}_t\leq w_t \leq \bar{w}_t\},\\
    \ubar{w}_t&= \max \left( \frac{v_{\mathrm{min}} - v_{0,t}}{T_s}, w_{\mathrm{min}} \right),\\
    \bar{w}_t&= \min \left( \frac{v_{\mathrm{max}} - v_{0,t}}{T_s}, w_{\mathrm{max}} \right).
\end{align}
Without this switching, the front car has unbounded velocity and robust invariant sets for the ego will be empty~(\cite{Kim_ACC_2019,Lefevre_ACC_2016}). 

\subsection{Robust \textit{M}-Step Hold Control Invariant Sets}
A robust $M$-step hold MPC for the ego car must be robust to the worst-case front car behavior: \textit{(i)} full breaking until reaching $v_{\mathrm{min}}$ and \textit{(ii)} full acceleration until reaching $v_{\mathrm{max}}$. Simultaneously robustifying against \textit{(i)} and \textit{(ii)} is sufficient to robustify against all front car behaviors. This task decomposition was used by~\cite{Kim_ACC_2019} and ~\cite{Lefevre_ACC_2016} to develop a technique for computing robust control invariant sets for~\eqref{eqn:acc_dyn}-\eqref{eqn:acc_W}. We adapt this technique to the $M$-step hold framework. The Multi-Parametric Toolbox 3.0 (MPT3) in MATLAB is used for all set computations~(\cite{MPT3}).

Offline, we compute the robust $M$-step hold control invariant set $\set{C}^M$ as a collection of individual robust $M$-step hold controllable sets for scenarios \textit{(i)} and \textit{(ii)}, denoted $\ubar{\set{C}}^M$ and $\bar{\set{C}}^M$, respectively (Alg.~\ref{alg:acc_sets}). 

Let $\mathbb{X}_v$ be a slice of $\set{X}$ for a fixed front car velocity:
\begin{equation}
    \mathbb{X}_{v}=\{x:x\in\set{X},~v_0=v\}
\end{equation}

To find $\ubar{\set{C}}^M$, we first compute the nominal $M$-step hold control invariant set for the front car at constant $v_{\mathrm{min}}$, denoted $\ubar{\mathbb{C}}^M_0$, using a fixed point algorithm of $\mathrm{Pre}^M(\set{S},0)$ initialized from $\mathbb{X}_{v_{\mathrm{min}}}$. The ego car is safe under scenario \textit{(i)} if it is robustly driven to $\ubar{\mathbb{C}}^M_0$. The states from which this is possible are found by computing robust $M$-step hold controllable sets to $\ubar{\mathbb{C}}^M_0$ iteratively for the full range of $v_0=v_{\mathrm{min}}\rightarrow v_{\mathrm{max}}$, denoted $\ubar{\mathbb{C}}^M_i$ with $i\in\{M,2M,\dots\}$. The collection $\{\ubar{\mathbb{C}}^M_0, \ubar{\mathbb{C}}^M_M, \dots\}$ represents $\ubar{\set{C}}^M$.

Scenario \textit{(ii)} is treated similarly. We find $\bar{\mathbb{C}}^M_0$, the nominal $M$-step hold control invariant set for the front car at constant $v_{\mathrm{max}}$, using $\mathrm{Pre}^M(\set{S},0)$ and $\mathbb{X}_{v_{max}}$. Robust $M$-step hold controllable sets to $\bar{\mathbb{C}}^M_0$ are found across $v_0=v_{\mathrm{max}}\rightarrow v_{\mathrm{min}}$, and the collection $\{\bar{\mathbb{C}}^M_0, \bar{\mathbb{C}}^M_M, \dots\}$ represents $\bar{\set{C}}^M$. Intersecting the two collections represents $\set{C}^M=\ubar{\set{C}}^M\cap\bar{\set{C}}^M$ (Fig.~\ref{fig:sets}). 

Online, we use Alg.~\ref{alg:acc_slices} at each solve of~\eqref{eqn:croc_min} to select a slice of $\set{C}^M$ based on the current $v_{0,t}$, denoted $\mathbb{C}^M_{v_0}$. Formulation~\eqref{eqn:croc_min} ensures recursive feasibility through a constraint on $x_{t+M|t}$, so the proper slice of $\set{C}^M$ to use at step $t$ is that closest to $v_{0,t+M|t}$. As this is unknown, we roll out $M$-steps of scenarios \textit{(i)} and \textit{(ii)} from $v_{0,t}$ to find an upper and lower bound on $v_{0,t+M|t}$. As $v_{0,t+M|t}$  is a continuous variable but Alg.~\ref{alg:acc_sets} computes slices at discrete values, we underestimate the lower bound to the closest $\ubar{\mathbb{C}}^M_{l}$, and overestimate the upper bound to the closest $\bar{\mathbb{C}}^M_{h}$. The intersection $\ubar{\mathbb{C}}^M_{l}\cap \bar{\mathbb{C}}^M_{h}$ is a conservative approximation of $\mathbb{C}^M_{v_0}$. We choose $\set{K}^M_{N-M}(\set{X}_N)=\ubar{\mathbb{C}}^M_{l}\cap \bar{\mathbb{C}}^M_{h}$ in~\eqref{eqn:croc_min}.

\begin{algorithm}[t]
\caption{Offline Computation of $\protect\ubar{\mathcal{C}}$ and $\bar{\mathcal{C}}$} \label{alg:acc_sets}
\textbf{Input:} System~\eqref{eqn:acc_dyn}-\eqref{eqn:acc_W}, $\set{X}$, $\mathbb{X}_{v_{\mathrm{min}}}$, $\mathbb{X}_{v_{\mathrm{max}}}$\\
\textbf{Output:} $\ubar{\mathcal{C}}^M$ and $\bar{\mathcal{C}}^M$
\begin{algorithmic}[1]
    \State $\Omega_0 \gets \mathbb{X}_{v_{\mathrm{min}}},~\Omega_{-1} \gets \emptyset$, $k \gets -1$
    \State \textbf{While} {$\Omega_{k+1} \neq \Omega_k$}
        \State \hspace{\algorithmicindent} $k \gets k + 1$
        \State \hspace{\algorithmicindent} $\Omega_{k+1} \gets \mathrm{Pre}^M(\Omega_k,0) \cap \Omega_k$
    \State $\ubar{\mathbb{C}}^M_0 \gets \Omega_{k+1}$, $v \gets v_{\mathrm{min}}, i \gets -M$
    \State \textbf{While} {$v < v_{\mathrm{max}}$}
         \State \hspace{\algorithmicindent}  $i \gets i + M$
         \State \hspace{\algorithmicindent} $v \gets v - MT_s w_{\mathrm{min}}$
         \State \hspace{\algorithmicindent} $\ubar{\mathbb{C}}^M_{i+M} \gets \mathrm{Pre}^M(\ubar{\mathbb{C}}^M_{i},w_{\mathrm{min}}) \cap \mathcal{X}$
    \State $\ubar{\mathcal{C}}^M \gets [\ubar{\mathbb{C}}^M_0, \ubar{\mathbb{C}}^M_M, \dots, \ubar{\mathbb{C}}^M_{i+M}]$
    \State $\Omega_0 \gets \mathbb{X}_{v_{\mathrm{max}}},~\Omega_{-1} \gets \emptyset$, $k \gets -1$
    \State \textbf{While} {$\Omega_{k+1} \neq \Omega_k$}
        \State \hspace{\algorithmicindent} $k \gets k + 1$
        \State \hspace{\algorithmicindent} $\Omega_{k+1} \gets \mathrm{Pre}^M(\Omega_k,0) \cap \Omega_k$
    \State $\bar{\mathbb{C}}^M_0 \gets \Omega_{k+1}$, $v \gets v_{\mathrm{max}}, i \gets -M$
    \State \textbf{While} {$v > v_{\mathrm{min}}$}
         \State \hspace{\algorithmicindent}  $i \gets i + M$
         \State \hspace{\algorithmicindent} $v \gets v - MT_s w_{\mathrm{max}}$
         \State \hspace{\algorithmicindent} $\bar{\mathbb{C}}^M_{i+M} \gets \mathrm{Pre}^M(\bar{\mathbb{C}}^M_{i},w_{\mathrm{max}}) \cap \mathcal{X}$
   \State $\bar{\mathcal{C}}^M \gets [\bar{\mathbb{C}}^M_0, \bar{\mathbb{C}}^M_M, \dots, \bar{\mathbb{C}}^M_{i+M}]$
\end{algorithmic}
\end{algorithm}

\begin{algorithm}[t] 
\caption{Online Computation of $\mathbb{C}^M_{v_{0}}$}\label{alg:acc_slices}
\textbf{Input:} System~\eqref{eqn:acc_dyn}-\eqref{eqn:acc_W}, $\ubar{\set{C}}^M$, $\bar{\set{C}}^M$, $v_{0}$\\
\textbf{Output:} $\mathbb{C}^M_{v_{0,t}}$
\begin{algorithmic}[1]
    \State $\ubar{v}_{0,t+M}\leftarrow v_{0,t}+MT_s w_{\mathrm{min}}$
    \State $\bar{v}_{0,t+M}\leftarrow v_{0,t}+MT_s w_{\mathrm{max}}$
    \State $l\leftarrow 
    \max_{i\in\{0,M,\dots\}} i~\text{s.t.}~\ubar{v}_{0,t+M}\geq v_{\mathrm{min}}-iT_s w_{\mathrm{min}}$
    \State $h\leftarrow \max_{i\in\{0,M,\dots\}} i~\text{s.t.}~\bar{v}_{0,t+M}\geq v_{\mathrm{max}}-iT_s w_{\mathrm{max}}$
    \State $\mathbb{C}^M_{v_{0}}\leftarrow\ubar{\mathbb{C}}^M_{l}\cap\bar{\mathbb{C}}^M_{h}$    
\end{algorithmic}
\end{algorithm}

\begin{figure}[t]
    \centering
    \includegraphics[width=1\linewidth]{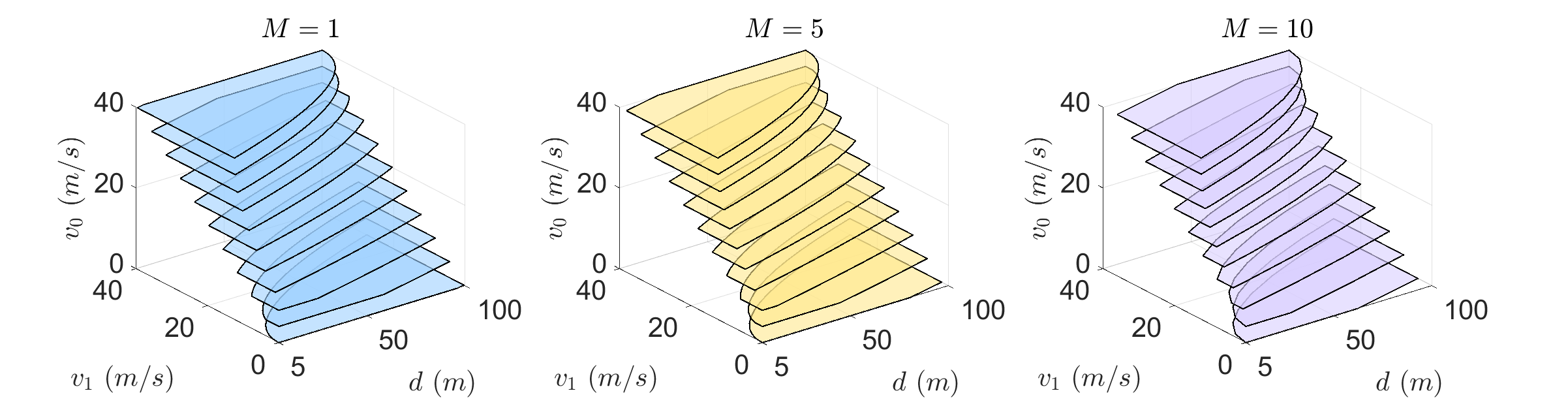}
    \caption{Slices of $\set{C}^1$, $\set{C}^5$, and $\set{C}^{10}$ 
    at discrete values of $v_{0}$.}
    \label{fig:sets}
\end{figure}

\subsection{Simulation Results}
\begin{table}[t]
\centering
\caption{Model and Control Parameters}
\label{tab:parameters}
\begin{tabular}{l l l r}
\toprule
$\{d_{\mathrm{min}},d_{\mathrm{max}}\}$ & distance bounds & m & \{5,100\} \\
$\{v_{\mathrm{min}},v_{\mathrm{max}}\}$ & velocity bounds & m/s & \{0,40\} \\
$\{u_{\mathrm{min}},u_{\mathrm{max}}\}$ & input bounds & m/s$^2$ & \{-4,4\} \\
$\{w_{\mathrm{min}},w_{\mathrm{max}}\}$ & uncertainty bounds & m/s$^2$ & \{-4,4\} \\
$T_s$ & sampling time & s & 0.1 \\
$M$ & $M$-step hold & - &$ \{1,5,10\}$\\
$N$ & MPC horizon & - & 10 \\
$Q$, $P$ & state cost & - & diag$(10,0,0)$ \\
$R$ & input cost & - & 1 \\
$[d_0,v_{1,0},v_{0,0}]^\top$ & initial state & - & $[70, 30, 25]^\top$\\
\bottomrule
\end{tabular}
\end{table}

\begin{figure}[t]
    \centering
    \includegraphics[width=1\linewidth]{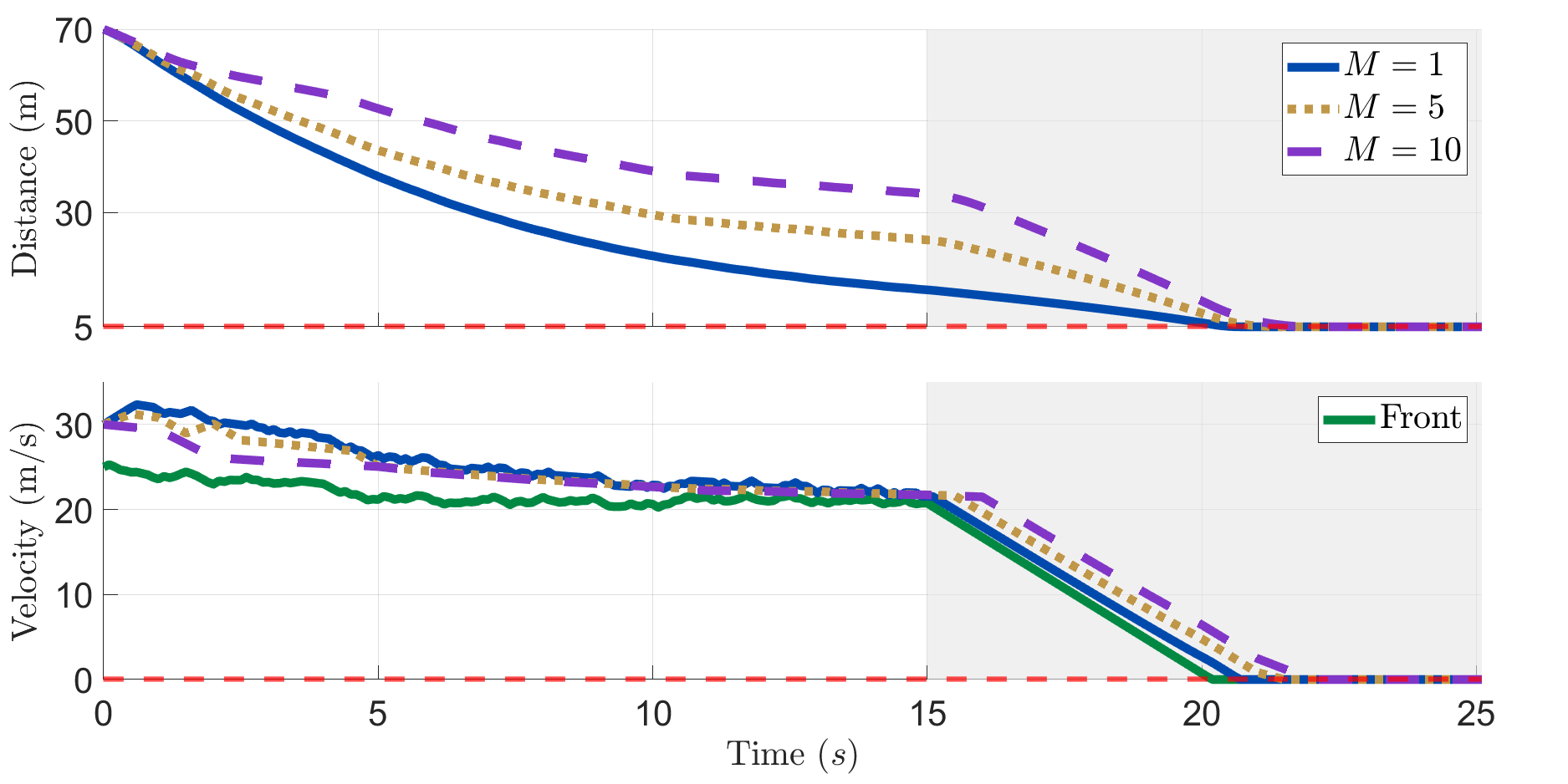}
    \caption{Larger $M$-step holds stay farther behind the front car to be robust to more open-loop uncertainty propagation, and all controllers safely stop the ego when the front car full-breaks until stopped ($\geq15$ s, shaded).}
    \label{fig:sim_break}
\end{figure}

\begin{figure}[t]
    \centering
    \includegraphics[width=1\linewidth]{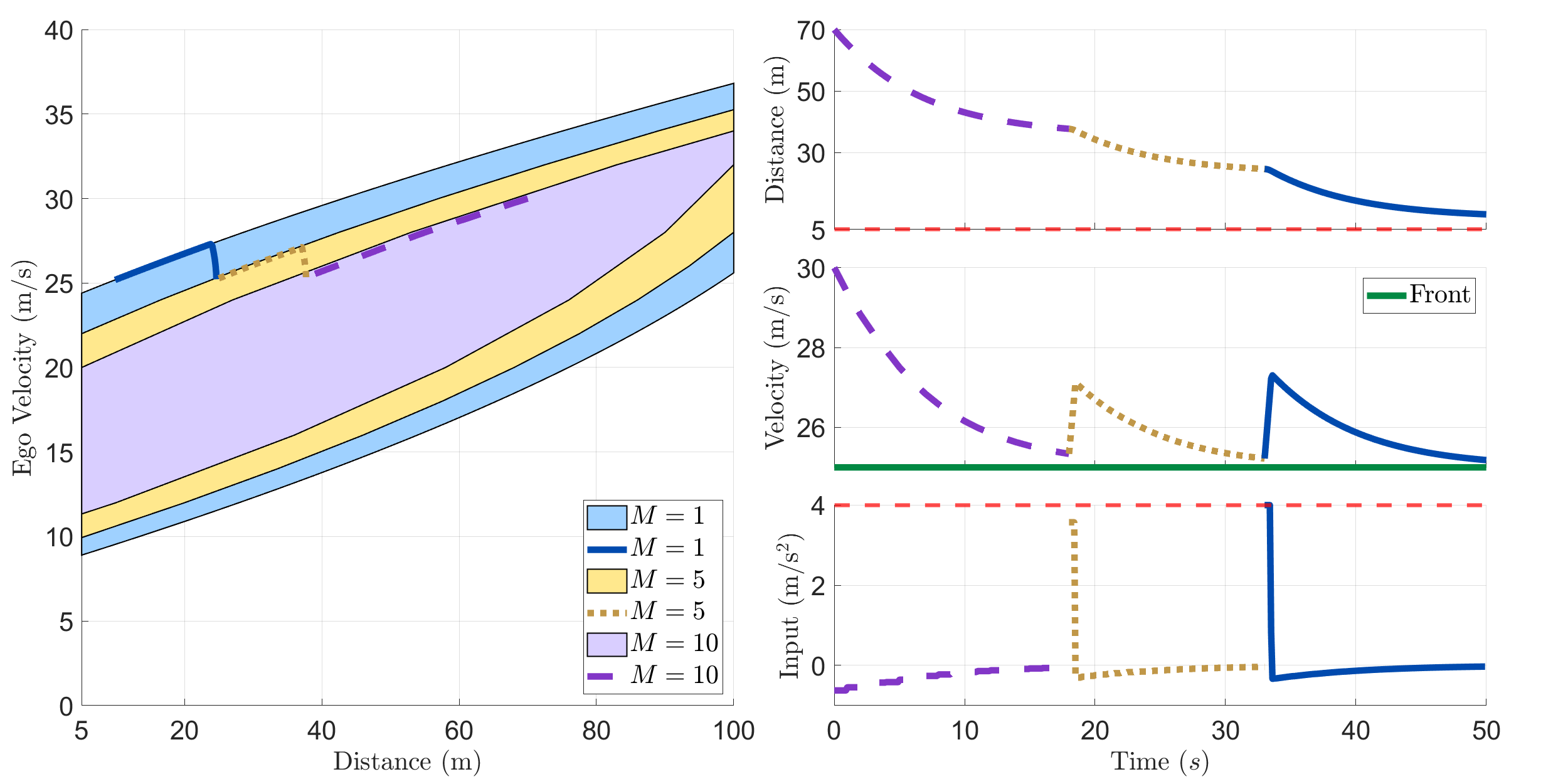}
    \caption{Online adaptation to smaller $M$-step holds expands the safe operating region to include smaller following distances. The slices $\mathbb{C}^M_{v_{0}}$ are nested for factors of $M$ for all $v_0$ ($25$ m/s shown).}
    \label{fig:sim_switch}
\end{figure}

We present two simulations (Figs.~\ref{fig:sim_break}-\ref{fig:sim_switch}) that showcase how the $M$-step hold intuitively affects control performance and robust constraint satisfaction. We simulate robust $M$-step hold MPC of~\eqref{eqn:acc_dyn}-\eqref{eqn:acc_W} that solves~\eqref{eqn:croc_min} with the parameters in Table~\ref{tab:parameters}, $\set{K}^M_N(\set{X}_N)=\mathbb{C}^M_{v_0}$, and $M\in\{1,5,10\}$. We choose a cost that heavily weights the following distance to highlight the effect of the safety constraints. 

Fig.~\ref{fig:sim_break} compares a $1$, $5$, and $10$-step hold MPC with a front car that applies a $w_t$ sampled uniformly at random from $\set{W}_t$ for $15$ seconds, then full breaks until stopped. In the first phase, larger $M$-step holds keep larger distances because they must be robust to more open-loop uncertainty propagation. In the second phase, all $M=\{1,5,10\}$ safely stop the ego car behind the front car. 

Fig.~\ref{fig:sim_switch} shows an $M$-step hold MPC that starts with $M=10$ and automatically adapts to $\hat{M}=5$ then $\hat{M}=1$ at a step $t$ s.t. mod$(t,M)=0$ if the distance has not changed by at least $1\%$ over the last second. The front car is simulated with a constant $v_{0,t}$ ($w_t=0$) so the trajectories may be visualized on a constant $\mathbb{C}^M_{v_0}$. Fig.~\ref{fig:sim_switch} shows that $\mathbb{C}^{10}_{v_0}\subseteq\mathbb{C}^5_{v_0}\subseteq\mathbb{C}^1_{v_0}$ for $v_0=25$ m/s, which we also observe for all $v_0\in[v_\mathrm{min},v_\mathrm{max}]$, demonstrating the intuitive evolution of robust $M$-step hold control invariant sets introduced by Thm.~\ref{thm:nest}. The trajectory, which here is exactly the nominal trajectory, is robustly within the proper $\mathbb{C}^M_{v_0}$ at all steps. When $M$ is decreased, the MPC applies an input to close the gap between the cars until reaching a new equilibrium between the cost and robust constraint satisfaction. This demonstrates how $M$ may be adapted online to adjust the  control performance and computational demand according to the task and environment. We will explore more sophisticated online adaptation algorithms in future work.

%

\section{Conclusion}\label{sec:conc}

This paper introduced robust $M$-step hold MPC, a control design for an uncertain discrete-time system model subject to a variable multi-step ($M$-step) input hold.
We demonstrated how to calculate $M$-step hold robust invariant sets and utilize them in the MPC design to robustly ensure recursive feasibility of the uncertain system. 
Our proposed framework was evaluated in a cruise control simulation example, which demonstrated how uncertainty propagation and input update rate affect control performance. 

\appendix
\section{\textit{M}-Step Hold Set Computations}\label{apx:sets}

\begin{algorithm}[t]
    \caption{Fixed-Point Alg. for $\set{C}^M_\infty$} \label{alg:fix_pt}
    \begin{algorithmic}[1]
        \Statex \hspace*{-\algorithmicindent} \textbf{Input:} System~\eqref{eqn:dyn_const}, $\set{X},~\set{W}$
        \Statex \hspace*{-\algorithmicindent} \textbf{Output: } $\set{C}^M_\infty$
        \State Let $\Omega_0 \gets \mathcal{X},~ \Omega_{-1}=\emptyset,~ k \gets -1$
        \State \textbf{While} $\Omega_{k+1}\neq\Omega_i$\textbf{:}
            \State \hspace{\algorithmicindent} $k \gets k+1$
            \State \hspace{\algorithmicindent} $\Omega_{k+1}\gets \mathrm{Pre}^M(\Omega_{k},\set{W})\cap\Omega_k$
        \State $\set{C}^M_\infty \gets \Omega_{k+1}$
    \end{algorithmic}
\end{algorithm}

Maximal invariant sets are computed via a fixed-point algorithm of precursor sets (Alg.~\ref{alg:fix_pt}). The set $\set{C}^M_\infty$ is found by using $\mathrm{Pre}^M(\set{S},\set{W})$ from Def.~\ref{def:msh_pre_set}. For polytopes $\set{X}=\{x:H_x x\leq h_x\}$, $\set{S}=\{x:H_s x\leq h_s\}$, and $\set{U}=\{u:H_u u\leq h_u\}$, $\mathrm{Pre}^M(\set{S},\set{W})$ is found with a polytope projection:
\begin{gather}
    \mathrm{Pre}^M(\set{S},\set{W}) = 
    \left\{x: \exists ~u~ \text{s.t.} ~\hat{H} 
    \begin{pmatrix} x \\ u \end{pmatrix} \leq \hat{h}\right\} \label{eqn:lti_pre}\\
    \hat{H} = \begin{bmatrix} 
    H_xA & H_xB \\ 
    \vdots & \vdots \\
    H_xA^{M-1} & H_x\left(\sum\limits_{j=0}^{M-2}{A^{j}B}\right)\\
    H_sA^M & H_s\left(\sum\limits_{j=0}^{M-1}{A^{j}B}\right)\\
    0 & H_u \end{bmatrix}~
     \hat{h}=\begin{bmatrix} \tilde{h}_{x,0} \\ \vdots \\ \tilde{h}_{x,M-2} \\ \tilde{h}_{s,M-1} \\ h_u \end{bmatrix} \nonumber\\
     \tilde{h}_{*,t,r}=\min_{w_k\in\set{W}_k}h_{*,t,r}-H_{*,r}\sum_{k=0}^t A^{t-k}Ew_{k} \nonumber 
\end{gather}
where $r$ refers to each row of $\tilde{h}$, $h$, and $H$, and $*=\{x,s\}$. 

The set $\set{O}^M_{\infty}$ is found using a modified precursor which considers~\eqref{eqn:dyn} in closed-loop with a chosen policy $\pi^M_t(\cdot)$:
\begin{align}
    &\mathrm{Pre}_\pi^M(\mathcal{S},\set{W}) = \big\{x_0:\exists \ u_0\in\mathcal{U} \text{ s.t. }\nonumber \\
    &~~~x_{t+1}=Ax_t+Bu_0+Ew_t,~x_{t+1}\in\set{X},~x_M\in\set{S},\nonumber\\
    &~~~u_0=\pi^M_0(\cdot),~\forall w_t\in\set{W}_t,~\forall t\in\{0,\dots,M-1\}\big\}.\label{eqn:pos_pre}
\end{align}
 For the same polytopic constraints and chosen policy $\pi^M_{t}(\cdot)=-Kx_{M\cdot\lfloor t/M\rfloor}$,~\eqref{eqn:pos_pre} is a polytope:
\begin{gather}
    \mathrm{Pre}_\pi^M(\cdot) = 
    \big\{x: ~\hat{H}_\pi 
    x \leq \hat{h}\big\} \label{eqn:lti_pos_pre}\\ \hat{H}_\pi=\begin{bmatrix}
    H_x(A-BK)\\
    \vdots\\ 
    H_x\left(A^{M-1}-\sum\limits_{j=0}^{M-2}{A^{j}BK}\right)\\
    H_s\left(A^{M}-\sum\limits_{j=0}^{M-1}{A^{j}BK}\right)\\
    -H_u K
    \end{bmatrix} \nonumber
\end{gather}
where $\hat{h}$ is unchanged from~\eqref{eqn:lti_pre}. 

\bibliography{root}

\end{document}